\documentclass{amsart}
\usepackage{amsaddr}
\usepackage{graphicx}%
\usepackage{multirow}%
\usepackage{amsmath,amssymb,amsfonts}%
\usepackage{amsthm}%
\usepackage[title]{appendix}%
\usepackage{xcolor}%
\usepackage{textcomp}%
\usepackage{manyfoot}%
\usepackage{booktabs}%
\usepackage{algorithm}%
\usepackage{algorithmicx}%
\usepackage{algpseudocode}%
\usepackage{listings}%
\usepackage{hyperref}
\usepackage[top=1in, bottom=1in, left=1in, right=1in]{geometry}

\newcommand{\Z}{\mathbb{Z}}

\newcommand{\C}{\mathcal{C}}

\newcommand{\F}{\mathbb{F}}

\let \tr \relax
\DeclareMathOperator{\tr}{Tr}

\newtheorem{theorem}{Theorem}[section]
\newtheorem{proposition}[theorem]{Proposition}
\newtheorem{lemma}[theorem]{Lemma}%

\newtheorem{remark}[theorem]{Remark}%
\newtheorem{example}[theorem]{Example}%
\newtheorem{definition}[theorem]{Definition}%

\begin{document}
\title{The Trace Dual of Nonlinear Skew Cyclic Codes}

\author[Bossaller]{Daniel Bossaller}
\address{Department of Mathematical Sciences, University of Alabama in Huntsville, 
Huntsville, AL, USA}
\email{daniel.bossaller@uah.edu}

\author[Herden]{Daniel Herden}
\address{Department of Mathematics, Baylor University, Waco, TX, USA}
\email{daniel\_herden@baylor.edu}

\author[Ruiz-Bola\~nos]{Indalecio Ruiz-Bola\~nos}
\address{Department of Mathematics, Texas State University, San Marcos, TX, USA}
\email{inda.ruiz@txstate.edu}

\thanks{The second author was supported by Simons Foundation grant MPS-TSM-00007788.}

\begin{abstract}
    Codes which have a finite field $\F_{q^m}$ as their alphabet but which are only linear over a subfield $\F_q$ are a topic of much recent interest due to their utility in constructing quantum error correcting codes. In this article, we find generators for trace dual spaces of different families of $\F_q$-linear codes over $\F_{q^2}$. In particular, given the field extension $\F_q\leq \F_{q^2}$ with $q$ an odd prime power, we determine the trace Euclidean and trace Hermitian dual codes for the general $\F_q$-linear cyclic $\F_{q^2}$-code. In addition, we also determine the trace Euclidean and trace Hermitian duals for general $\F_q$-linear skew cyclic $\F_{q^2}$-codes, which are defined to be left $\F_q[X]$-submodules of $\F_{q^2}[X;\sigma]/(X^n-1)$, where $\sigma$ denotes the Frobenius automorphism and $\F_{q^2}[X;\sigma]$ the induced skew polynomial ring.
\end{abstract}
\keywords{skew-cyclic codes, skew polynomials, $\F_q$-linear $\F_{q^m}$-codes}
\subjclass[2020]{94B05, 94B15, 11T71, 16S36}
\maketitle

\section{Introduction}\label{intro}
Additive codes were first investigated in 1971 by Delsarte in \cite{delsarte1971}. These are codes which are closed under the addition operation of the alphabet $\F_{q^m}$ but fail to be closed under the scalar multiplication by $\F_{q^m}$. In the present article, we consider additive codes which are closed under scalar multiplication by some subfield $\F_q$. In particular, these codes form $\F_q$-linear subspaces of $\F_{q^m}^n$, which we will refer to as \emph{$\F_q$-linear $\F_{q^m}$-codes}, or as ``nonlinear'' codes when the alphabet and subfield are understood. In the realm of quantum error-correcting codes, there is much interest in nonlinear codes whose alphabet is a degree $2$ extension of $\F_q$, i.e., $\F_{q^2}$, see for example the construction of Ashikhmin and Knill in \cite{Ashikhmin2001}. Later work from Ketkar et al. \cite{Ketkar2006} established an intimate connection between the existence of quantum error-correcting codes and $\F_q$-linear $\F_{q^2}$-codes equipped with an appropriate inner product.

Another direction for the construction of quantum error-correcting codes is the Calderbank-Shor-Steane (CSS) Construction which uses two codes over $\F_q$, $\mathcal C_1$ and $\mathcal C_2$, such that $\C_2^\perp \subseteq \C_1$. For full details of the construction, see \cite{Calderbank1998}. Further work, especially ``Quantum Construction X'' introduced by Lison{\v e}k and Singh \cite{lisonek2013}, established the ``hull'' $\C \cap \C^\perp$ of a code $\C$ as an important property to study. As noted in the article by Choi et al. \cite{Choi2023}, nonlinear codes are well positioned to have an ``additive complementary dual,'' where $\C \cap \C^{\perp} = \{0\}$.

Among the most ubiquitous of codes in use today are cyclic codes and their generalizations. Their ubiquity is due in part to their rich algebraic structure. A code $\C$ is said to be \emph{cyclic} if, whenever $(c_0, c_1, \ldots, c_{n-1}) \in \C$, then $(c_{n-1}, c_0, \ldots, c_{n-2}) \in \C$. Algebraically, given a finite field, $\F_q$, define the quotient ring
\[P_n = \F_q[X]/(X^n - 1),\] where $(X^n - 1)$ denotes the principal ideal of $\F_q[X]$ generated by the polynomial $X^n - 1$. The (two-sided) ideals of $P_n$ are isomorphic to the cyclic codes of length $n$ over $\F_q$. These ideals are principal, so the cyclic codes of length $n$ over $\F_q$ are generated by the polynomials $g(x) \in \F_q[x]$ which divide $x^n - 1$. Moreover the Euclidean dual of a cyclic is again a cyclic code generated by the reciprocal polynomial of $(x^n - 1)/g(x)$. It was shown in \cite{Verma-Sharma:2024} that $\F_q$-linear, cyclic $\F_q^2$ codes and their trace-Euclidean and trace-Hermitian duals are generated by two polynomials. 

Many generalizations of cyclic codes have been studied; among those of most relevance to the present work are the additive cyclic codes first studied in a more generalized framework by Bierbrauer in 2012, \cite{Bierbrauer2012}, and in a more specific setting by Huffman in \cite{Huffman:2010} and \cite{Huffman:2013}. In those articles, Huffman considers $\F_q$-linear $\F_{q^t}$-codes, employing trace operators and trace inner products to define their trace dual spaces. He then enumerates the cyclic self-orthogonal and cyclic self-dual codes with respect to these trace inner products. In \cite{Shi2008}, the authors investigate additive cyclic and additive complementary dual codes of odd length over $\F_4$ with respect to these trace inner products. In addition to the above articles by Huffman, the reader is encouraged to consult \cite{Bhunia2025} for a wide study of additive codes whose alphabets come from finite fields.

A separate direction of generalization for cyclic codes comes in the form of skew cyclic codes. These codes were first defined by Boucher, Geiselmann, and Ulmer in \cite{Boucher-Geiselmann-Ulmer:2007} as the left ideals of the ring $\F_{q^m}[X; \sigma]/(X^n - 1)$, where $\sigma: \F_{q^m} \to \F_{q^m}$ is an automorphism with fixed field $\F_{q}$ and $\sigma^n=\operatorname{id}$. This has grown into a large field of study; please consult \cite{Gluesing-Luerssen:2021} for a survey. In \cite{BHR_1}, the authors of the present article studied codes which are simultaneously $\F_q$-linear $\F_{q^m}$-codes and skew cyclic codes. As the degree of the field extension plays a key role in the classification of such codes, we will choose to use ``$\F_{q}$-linear $\F_{q^m}$-codes'' as opposed to ``additive codes'' to emphasize this fact.

In \cite{Verma-Sharma:2024}, for $q$ any odd prime power, the authors study additive cyclic codes over $\F_{q^2}$, which possess the same structural properties as the $\F_q$-linear cyclic $\F_{q^2}$-codes defined in \cite{BHR_1}. They present a general characterization of these codes by constructing explicit generators. In addition, for cyclic codes of length $n$ with $\gcd(n,q)=1$, they derive generators for the trace dual spaces of a specific family of these codes, though the problem was left unsolved in the general case.

In Section~\ref{sec:NonlinearCyclic} of this paper, we introduce terminology and notation that we use in the remainder of the article. In Section~\ref{sec:dualcycliccodes} we demonstrate that the hypothesis of Theorem 4.8 in \cite{Verma-Sharma:2024} corresponds to a collection of codes whose trace Euclidean dual space has a straightforward description, see \cite[Corollary 4.2]{Verma-Sharma:2024}. Additionally, we derive generators for the trace Euclidean and trace Hermitian duals of a general $\F_q$-linear cyclic $\F_{q^2}$-code while also dropping the condition $\gcd(n,q)=1$ on the length $n$ of the code. 

Finally, in Section~\ref{sec:dualskewcyclicodes}, we study $\F_q$-linear skew cyclic $\F_{q^2}$-codes of length $n$. We find the general structure of these codes by providing two generators as left $\F_q[X]$-submodules of $\F_{q^2}[X;\sigma]/(X^n-1)$. Then, we derive generators for the trace Euclidean and trace Hermitian duals of a general $\F_q$-linear skew cyclic $\F_{q^2}$-code of length $n$.

\section{\texorpdfstring{$\F_q$}{Fq}-linear Cyclic Codes over \texorpdfstring{$\F_{q^2}$}{Fq2}}\label{sec:NonlinearCyclic}

Consider a field extension $\F_q\leq \F_{q^2}$ of extension degree $2$, where $q$ is an odd prime power. It is well known that there is a nontrivial automorphism of $\F_{q^2}$ which fixes $\F_q$, namely the \emph{Frobenius automorphism} $\sigma: \F_{q^2} \to \F_{q^2}$ given by $\sigma(x) := x^q$. Moreover, we may define the \emph{trace} $\operatorname{Tr}: \F_{q^2} \to \F_q$ by $\operatorname{Tr}(x):= x+ x^q$ for all $x\in \F_{q^2}$.

The following lemma is a straightforward application of the theory of finite fields of odd prime characteristic. We refer the reader to \cite{1997-lidl-nied} or \cite{Hou2018} for a proof and further discussion. 

\begin{lemma}
There exists $\gamma \in \F_{q^2}\setminus \F_q$ such that $\tr(\gamma)=0$. Thus, in particular, $\F_{q^2}$ may be written as
\[\F_{q^2} = \F_q + \gamma \F_q.\]
\end{lemma}

\noindent As a result of this lemma, the polynomial ring $\F_{q^2}[X]$ can be decomposed as
\[\F_{q^2}[X] = \F_q [X]+ \gamma \F_q[X]\]
by rewriting each coefficient in terms of the basis $\{1, \gamma\}$ of $\F_{q^2}$ over $\F_q$.

In addition, we note the following basic properties of $\gamma$ for later reference.

\begin{lemma}\label{gammasquare}
We have $\gamma^2 \in \F_q$. In particular, $\tr(\gamma^2)= 2\gamma^2 \ne 0$. Moreover, $\gamma^{q+1}=-\gamma^2\in \F_q$.
\end{lemma}

\begin{proof}
 $\tr(\gamma)=0$ implies $\gamma^q=-\gamma$. Since $\gamma \ne 0$, we have $\gamma^{q-1}=-1$ and $\left(\gamma^2\right)^{q-1}=1.$ Thus $\gamma \in \F_{q^2} \setminus \F_q$, but $\gamma^2 \in \F_q$, and $\tr(\gamma^2)= 2\gamma^2 \ne 0$ follows. Furthermore, $\gamma^{q+1}=\gamma^q \gamma = (-\gamma) \cdot \gamma = -\gamma^2.$
\end{proof}

We continue with a quick review of cyclic codes.

\begin{definition}
    An \emph{$\F_q$-linear cyclic $\F_{q^2}$-code of length $n$} is an $\F_q$-linear subspace $\C$ of $\F_{q^2}^n$ that satisfies
\[(c_0, c_1, \ldots, c_{n-1}) \in \C \implies (c_{n-1},c_0,\ldots,c_{n-2})\in \C.\]    
\end{definition}

In order to study dual codes of $\F_q$-linear cyclic $\F_{q^2}$-codes of length $n$, we will make use of polynomials with coefficients in $\F_{q^2}$ and some associated polynomial algebras.
Denote
\[R_n := \F_{q^2}[X]/\left( X^n-1\right),\qquad\ \]
\[ P_n:=\F_{q}[X]/\left( X^n-1\right)\subseteq R_n,\]
where $\F_{q^2}[X]$ and $\F_{q}[X]$ are the rings of polynomials with coefficients in $\F_{q^2}$ and $\F_q$, respectively. 

\begin{definition}
    Let $f(X)= \sum_{i=0}^m f_i X^i \in \F_{q} [X]$ with $f_m\ne 0$, $m\geq0$. The \emph{reciprocal polynomial} of $f(X)$ is $f^*(X):= \sum_{i=0}^m f_i X^{m-i}$.
   We also define $\widehat{f}(X):= f(X^{n-1})$, and for the zero polynomial $f(X)=0$ we set $f^*(X)= \widehat{f}(X) =0$. For $\deg(f) \le n$, notice that
   \[ \widehat{f}(X) +(X^n-1) = X^{n-\deg f(X)} f^*(X) +(X^n-1) \]
   in $R_n$.
\end{definition}

We note the following two lemmas about reciprocal polynomials which will be used in the proof of Proposition \ref{prop:3.3}.

\begin{lemma}\label{lem:singlestar}
    Let $f(X),g(X)\in \F_q[X]$. Then there exist $k,\ell\in \Z$ such that $\left(f+g\right)^* =X^k f^* +X^\ell g^*$.
\end{lemma}

\begin{proof}
     We will focus on the general case where all of $f$, $g$, and $f+g$ are nonzero polynomials. The cases where $f(X)=0$ or $g(X)=0$ or $(f+g)(X)=0$ are straightforward.

    Consider $f(X) =\sum_{i=0}^m f_i X^i$ and $g(X) =\sum_{i=0}^n g_iX^i$, where $f_m, g_n\ne 0$. Then, $f^*(X) =\sum_{i=0}^m f_i X^{m-i}$ and $g^*(X) =\sum_{i=0}^n g_iX^{n-i}$. Without loss of generality, assume that $m \geq n$. Define $g_i =0$ for $i=n+1, \ldots, m$ to write $g(X) =\sum_{i=0}^m g_iX^i$.
    
    Define $D = \min \left\{ i \geq 0 : f_{m-i}+  g_{m-i} \ne 0 \right\}$. As $f+g \ne 0$, $D$ is well-defined. Then,
 \[ f(X)+ g(X)= \sum_{i=0}^m \left(f_i + g_i\right)X^i = \sum_{i=0}^{m-D} \left(f_i + g_i\right)X^i.\]   
 The reciprocal polynomial of $(f+g)$ is
 \begin{align*}
     \left(f+g\right)^* &= \sum_{i=0}^{m-D} \left(f_i + g_i\right) X^{m-D -i} \\
     &= \sum_{i=0}^{m} \left(f_i + g_i\right) X^{m-D -i}, \mbox{ because }f_i + g_i=0\mbox{ for }i=m-D+1, \ldots, m, \\
     & = X^{-D} \sum_{i=0}^{m} \left(f_i + g_i\right) X^{m -i}\\
     & = X^{-D} \sum_{i=0}^{m} f_i X^{m -i} + X^{-D} \sum_{i=0}^{m} g_i X^{m -i} \\
     & = X^{-D} f^* (X) + X^{-D} \sum_{i=0}^{n} g_i X^{m -i}, \mbox{ because }g_i=0\mbox{ for }i=n+1, \ldots, m, \\
     & =  X^{-D} f^* (X) + X^{-D} X^{m-n} \sum_{i=0}^{n} g_i X^{n -i} \\
     & = X^{-D} f^* (X) + X^{m-n-D} g^*(X). \qedhere
 \end{align*}
\end{proof}

\begin{lemma}\label{lem:doublestar}
    Let $f(X)= \sum_{i=0}^m f_i X^i\in \F_q[X]$ a nonzero polynomial and $k=\min \left\{ i: f_i \ne 0 \right\}$. Then $X^k f^{**}(X)= f(X)$.
\end{lemma}
\begin{proof}
    Since $f_0=\ldots =f_{k-1}=0$, we have that
    \[ f^* (X)= \sum_{i=0}^m f_i X^{m-i}=\sum_{i=k}^m f_i X^{m-i}=  \sum_{j=0}^{m-k} f_{j+k} X^{m-k-j}. \]
Thus, $\deg f^{**}(X)= m-k$ and
\[X^k f^{**}(X)= X^k \sum_{j=0}^{m-k} f_{j+k} X^{j}=  \sum_{j=0}^{m-k} f_{j+k} X^{j+k}=\sum_{i=k}^m  f_i X^i = f(X). \qedhere \]
\end{proof}

We will examine four natural inner products on our $\F_q$-linear cyclic codes over $\F_{q^2}$. Two of these, the Euclidean and Hermitian inner products are widely studied in the context of linear codes. As with \cite{Verma-Sharma:2024}, we also define trace Euclidean and trace Hermitian inner products which are more appropriate for nonlinear codes.

Given 
$\boldsymbol{x},\boldsymbol{y}\in \F_{q^{2}}^n$, we define the $\emph{Euclidean inner product}$ $\langle \boldsymbol{x},\boldsymbol{y}\rangle$ by
\[\langle \boldsymbol{x},\boldsymbol{y}\rangle=\boldsymbol{x}\cdot \boldsymbol{y} = \sum_{i=1}^{n}x_i y_i \in \F_{q^{2}}. \]
We also define the \emph{Hermitian inner product}  $\langle\boldsymbol{x},\boldsymbol{y}\rangle_H$ by
\[\langle \boldsymbol{x},\boldsymbol{y}\rangle_H=\sum_{i=1}^{n}x_i y_i^{q}\in \F_{q^2}.  \]

Any element of $R_n =\F_{q^2}[X]/\left( X^n-1\right)$ can uniquely be written as $f(X)+(X^n-1)$ for a suitable representative $f(X)=\sum_{i=0}^{n-1}f_i X^i\in \F_{q^2}[X]$, and we will often conflate $f(X)+(X^n-1)\in R_n$ with its representative $f(X)=\sum_{i=0}^{n-1}f_i X^i\in \F_{q^2}[X]$ for ease of notation. Given $f(X)+(X^n-1)$, $g(X)+(X^n-1)\in R_n$ with $f(X)=\sum_{i=0}^{n-1}f_i X^i$ and $g(X)=\sum_{i=0}^{n-1}g_iX^i$, we define the \emph{Euclidean inner product} as
\[f(X) \star g(X):= \langle (f_0, f_1,\ldots, f_{n-1}), (g_0, g_1, \ldots, g_{n-1})  \rangle \]
and the \emph{Hermitian inner product} as
\[f(X) \bullet g(X):= \langle (f_0, f_1,\ldots, f_{n-1}), (g_0, g_1, \ldots, g_{n-1})  \rangle_H. \]

The Euclidean inner product is $\F_{q^2}$-bilinear while the Hermitian inner product is $\F_{q^2}$-linear in $f(X)$  and only $\F_{q}$-linear in $g(X)$. For $f(X)+(X^n-1)$, $g(X)+(X^n-1)\in R_n$ we also define two $\F_q$-bilinear inner products: the \emph{trace Euclidean inner product} is
\[f(X) \circledast  g(X):= \tr \left( f(X) \star g(X) \right) \]
and the \emph{trace Hermitian inner product} is
\[f(X)  \boxdot g(X):= \tr \left( f(X) \bullet g(X) \right).\]

\begin{proposition} \label{prop 2.6}
    Given $f(X)+(X^n-1)$, $g(X)+(X^n-1)$, $q(X)+(X^n-1)\in P_n=\F_{q}[X]/\left( X^n-1\right)$ and $\alpha \in \F_{q^2}$, the previously defined inner products satisfy the following properties which will be used for Theorems~\ref{main cyclic} and~\ref{main skew cyclic}.
    \begin{align}
            f(-X)\star g(-X) & = f(X)\star g(X),  \label{lem:flipsings}\\
            \left( f(X)q(X) \right)\star g(X) & = f(X)\star  \left( g(X)\widehat{q}(X) \right), \mbox{ see \cite[Lemma 4.6]{Verma-Sharma:2024},} \label{lem:passtheq}\\
            f(X)\circledast g(X)&= 2 \left( f(X)\star g(X)\right), \notag\\
            f(X)\circledast \alpha g(X)&= \alpha f(X)\circledast  g(X)= \tr(\alpha)\left( f(X)\star g(X)\right) , \notag \\
             \alpha f(X)\circledast \alpha g(X)&=  \tr(\alpha^2)\left( f(X)\star g(X)\right) , \notag \\
             f(X) \bullet  g(X)&= f(X)\star  g(X),\notag \\
             f(X) \boxdot  g(X)&= 2 \left( f(X)\star  g(X)\right),\notag \\ 
            f(X) \boxdot  \alpha g(X)&= \alpha f(X) \boxdot   g(X)= \tr(\alpha)\left( f(X)\star  g(X)\right) , \notag \\
             \alpha f(X) \boxdot  \alpha g(X)&=  \tr(\alpha^{q+1})\left( f(X)\star  g(X)\right) \notag.  
    \end{align}
\end{proposition}

Let $\C\subseteq R_n$ be an $\F_q$-linear cyclic $\F_{q^2}$-code of length $n$, i.e.,  $\C$ is an $\F_{q}[X]$-submodule of $R_n$. The \emph{Euclidean dual}, \emph{trace Euclidean dual}, \emph{Hermitian dual}, and \emph{trace Hermitian dual} of $\C$, respectively, are defined as
\[
\begin{gathered}
\C^{\perp_E}=\left\{f(X) \in R_{n}: f(X) \star c(X)=0 \text { for all } c(X) \in \C\right\}, \\
\C^{\perp_{T E}}=\left\{f(X) \in R_{n}: f(X) \circledast c(X)=0 \text { for all } c(X) \in \C\right\}, \\
\C^{\perp_H}=\left\{f(X) \in R_{n}: f(X) \bullet c(X)=0 \text { for all } c(X) \in \C\right\}, \\
\C^{\perp_{T H}}=\left\{f(X) \in R_{n}: f(X) \boxdot c(X)=0 \text { for all } c(X) \in \C\right\} .
\end{gathered}
\]

    Given $\C$ an $\F_q$-linear cyclic $\F_{q^2}$-code of length $n$ and $\F_q$-vector space dimension $k$, let $\langle \C\rangle$ denote the $\F_{q^2}$-linear closure of $\C$ and $k^*$ its $\F_{q^2}$-vector space dimension. If $b_1, \ldots, b_{k^*}$ is a basis of $\langle \C\rangle $ over $\F_{q^2}= \F_q +\gamma \F_q$, then $b_1, \ldots, b_{k^*},\gamma b_1, \ldots, \gamma b_{k^*}$ is a basis of $\langle \C\rangle$ over $\F_q$. Therefore, $\frac{k}{2}\leq k^* \leq k$.

\begin{example}
    Let $f(X) := \sum_{i=0}^{n-1}X^i\in \F_q[X]$. Then $\C= \langle f(X)\rangle_{\F_q}\subseteq R_n$, the generated $\F_q$-subspace, is an $\F_q$-linear cyclic $\F_{q^2}$-code of length $n$ with $k=1$, $\langle \C\rangle=\langle f(X)\rangle_{\F_{q^2}}$, and $k^* = 1=k$. On the other hand, if $\mathcal{C}= \langle f(X), \gamma f(X)\rangle_{\F_q}$, then $k=2$, $\langle \mathcal{C}\rangle=\langle f(X)\rangle_{\F_{q^2}}=\C$, and $k^* = 1=\frac{k}{2}$.
\end{example}

\begin{remark}
The Euclidean inner product $\star$ is $\F_{q^2}$-bilinear. In particular, we have $\C^{\perp_E}=\langle \C \rangle^{\perp_E}$ and
\[ \dim_{\F_{q^2}} \C^{\perp_E}=   \dim_{\F_{q^2}} \langle \C \rangle^{\perp_E} = n - \dim_{\F_{q^2}} \langle \C\rangle = n- k^*.  \]

Since $\F_{q^2}= \F_q + \gamma \F_q$, we have that $\dim_{\F_{q}} \C^{\perp_E}=2\left( n- k^* \right)= 2n-2k^*.$ Notice that, unlike stated in \cite[Section 2, p.1597]{Verma-Sharma:2024}, $\dim_{\F_{q}} \C + \dim_{\F_{q}} \C^{\perp_E}=2n$ if and only if $k=2k^*$. 
%
%
%
\end{remark}
%

The Lemma below describes the generators of linear cyclic $\F_q$-codes and of their Euclidean duals. A proof can be found in \cite[Section 4.2]{Huffman-Pless:2003}.
\begin{lemma}\label{euclideandual}
    Let $\C\subseteq P_n$ be a linear cyclic $\F_q$-code, Then $\C=\langle f(X)\rangle$ with $f(X)\mid X^n-1$. Moreover, if $g(X)=\dfrac{X^n-1}{f(X)}$ then $\C^{\perp_{E}}= \langle g^*(X)\rangle$.
\end{lemma}


Let $\C$ be an $\F_q$-linear cyclic $\F_{q^2}$-code of length $n$. According to \cite{Verma-Sharma:2024}, there exist $w(X)$, $\ell(X)$, $f(X)$, $g(X)$, $q(X)\in \F_q[X]$ such that $\deg(q)<n$,
    \[X^n-1= w(X)\ell(X)f(X) g(X)\quad \mbox{ and, }\quad \C = \langle w(X) f(X)+ \gamma q(X), \gamma w(X)g(X)\rangle. \]

If also $\gcd(n,q)=1$, then $w(X)$ divides $q(X)$ and there exists $q_1(X)\in \F_q[X]$ such that
\[ \C = \langle w(X) f(X)+ \gamma w(X)q_1(X), \gamma w(X)g(X)\rangle.  \]

\section{\texorpdfstring{The Trace Dual of $\F_q$-linear $\F_{q^2}$ Cyclic Codes}{The Trace Dual of Fq-linear Fq2 Cyclic Codes}}\label{sec:dualcycliccodes} 
Next we discuss the trace Euclidean dual code of a special family of $\F_q$-linear cyclic $\F_{q^2}$-codes. This result is identical with \cite[Corollary 4.2]{Verma-Sharma:2024},
and we exhibit the proof mainly in preparation of the significantly more complex calculations necessary for proving Theorems \ref{main cyclic} and  \ref{main skew cyclic} later. Note, in particular, that we intentionally choose explicit calculations and direct methods over any dimensional arguments such as using the fact $\dim_{\F_q}\C^{\perp_{TE}}= 2n-k$.

\begin{theorem} \label{verma cor}
    Let $w(X), \ell(X), f(X),g(X)\in \F_q[X]$ such that 
    \[X^n-1= w(X) \ell(X) f(X)g(X).\] 
    
    Define the $\F_q$-linear cyclic $\F_{q^2}$-codes \[\C= \langle 
 w(X) f(X), \gamma w(X)g(X) \rangle\quad \mbox{ and }\quad \mathcal{D}= \langle 
 \ell^*(X) g^*(X), \gamma \ell^*(X)f^* (X) \rangle.\]  Then $\C^{\perp_{TE}}=\mathcal{D}.$
\end{theorem}

\begin{proof}
    To show $\C^{\perp_{TE}}\subseteq \mathcal{D}$, take $c(X)+ \gamma d(X)\in \C^{\perp_{TE}}$ with $c(X),d(X)\in \F_q[X]$. Then, for any $a(X), b(X)\in \F_q[X]$ we have that 
    \[ \left(a(X)w(X)f(X)+b(X)\gamma w(X)g(X)\right)\circledast \left( c(X) + \gamma d(X) \right) =0.\]

    In particular, if $a(X)=0$, then 
    \begin{align*}
        0 &= b(X)\gamma w(X)g(X)\circledast  c(X) + b(X)\gamma w(X)g(X)\circledast \gamma d(X) \\
        &= \tr{(\gamma)}[b(X) w(X)g(X)\star  c(X)] +\tr{(\gamma^2)}[b(X) w(X)g(X)\star  d(X)]\\
        &= \tr{(\gamma^2)} [b(X) w(X)g(X)\star  d(X)].
    \end{align*}
    Thus, $d(X)\in \langle \left[ \left(X^n-1\right)/\left(w(X) g(X)\right)\right]^* \rangle= \langle  \ell^*(X) f^*(X) \rangle$ by Lemma \eqref{euclideandual}.

    Besides, if $b(X)=0$, then 
    \begin{align*}
        0 &= a(X)w(X)f(X)\circledast  c(X) + a(X)w(X)f(X)\circledast \gamma d(X) \\
        &= a(X)w(X)f(X)\circledast  c(X) +\tr{(\gamma)} [a(X)w(X)f(X)\circledast d(X)] \\
        &= a(X)w(X)f(X)\circledast  c(X).
    \end{align*}
    Thus, $c(X)\in \langle \left[ \left(X^n-1\right)/\left(w(X) f(X)\right)\right]^* \rangle= \langle  \ell^*(X) g^*(X) \rangle$, and $\C^{\perp_{TE}}\subseteq \mathcal{D}$ follows. From the above calculations it is now easy to also verify $\mathcal{D}\subseteq \C^{\perp_{TE}}$.
\end{proof}

\begin{remark}
    Our next proposition shows that the conditions of \cite[Theorem 4.8]{Verma-Sharma:2024} are only satisfied by codes of the form $\C= \langle 
 w(X) f(X), \gamma w(X)g(X) \rangle$. In particular, \cite[Theorem 4.8]{Verma-Sharma:2024} reduces to the special case as described in \cite[Corollary 4.2]{Verma-Sharma:2024} and Theorem \ref{verma cor}. The mathematics necessary to demonstrate this is of quite some interest in its own rights.
\end{remark}

\begin{proposition}\label{prop:3.3}
    Let $w(X), \ell(X), f(X),g(X)\in \F_q[X]$ such that 
    \[X^n-1= w(X) \ell(X) f(X)g(X).\]
\begin{itemize}
\item[$(a)$] Let $q(X)\in \F_q[X]$.
    If $\widehat{q}(X)\equiv g^*(X) \mod \left( w^*(X)f^*(X) g^*(X) \right)$, then $g(X)$ divides $q(X)$, and 
    \[ \langle 
 w(X) f(X)+ \gamma w(X)q(X), \gamma w(X)g(X) \rangle = \langle 
 w(X) f(X), \gamma w(X)g(X) \rangle. \]
\item[$(b)$] Let $r(X)\in \F_q[X]$. If $r(X)\equiv f^*(X) \mod \left( w^*(X)f^*(X) g^*(X) \right)$, then $f^*(X)$ divides $r(X)$, and
    \[ \langle 
 \ell^*(X) g^*(X)+ \gamma \ell^*(X)r(X), \gamma \ell^*(X)f^*(X) \rangle = \langle 
 \ell^*(X) g^*(X), \gamma \ell^*(X)f^*(X) \rangle . \]
\end{itemize}
\end{proposition}

\begin{proof}
    $(a)$ Since $X^n-1= w(X) \ell(X) f(X)g(X)$, we have that $w(X)$, $\ell(X)$, $f(X)$, and $g(X)$ have a nonzero constant term. Then, by Lemma \ref{lem:doublestar}, we have that $w^{**}(X)=w(X)$, $\ell^{**}(X)=\ell(X)$, $f^{**}(X)=f(X)$, and $g^{**}(X)=g(X)$.

    Notice that
    \begin{align*}
        &\widehat{q}(X)\equiv g^*(X) \mod \left( w^*(X)f^*(X) g^*(X) \right)\\
        & \implies X^{n-\deg q(X)}q^*(X)  \equiv g^*(X) \mod \left( w^*(X)f^*(X) g^*(X) \right) \\
        & \implies X^{n-\deg q(X)}q^*(X)  - g^*(X) = a(X) w^*(X)f^*(X) g^*(X) \mbox{ for some }a(X)\in \F_q[X]\\
 & \implies ( X^{n-\deg q(X)}q^*(X)  - g^*(X))^* = \left(a(X) w^*(X)f^*(X) g^*(X)\right)^* \\
 & \implies X^k q^{**}(X)- X^l g^{**}(X) = a^*(X) w^{**}(X)f^{**}(X) g^{**}(X) \mbox{ for some }k,l\in \Z, \mbox{ by Lemma }\ref{lem:singlestar} \\
 & \implies X^k q(X)- X^l g(X) = a^*(X) w(X)f(X) g(X) \mbox{ for some }k,l \in \Z \\
 & \implies X^k q(X) = ( X^l + a^*(X) w(X)f(X) ) g(X).
    \end{align*}
Since $\gcd(g(X), X^k)=1$, we have that $g(X)\mid q(X)$.

$(b)$ Suppose that $r(X)\equiv f^*(X) \mod \left( w^*(X)f^*(X) g^*(X) \right)$. Then there exists $a(X)\in \F_q[X]$ such that $r(X)- f^*(X)= a(X)w^*(X)f^*(X) g^*(X)$. It follows that $r(X)= f^*(X)\left( 1+ a(X)w^*(X) g^*(X)\right)$, hence $f^*(X)\mid r(X)$.
\end{proof}

We continue with a closed description of the trace dual $\C^{\perp_{TE}}$ of a general $\F_q$-linear cyclic $\F_{q^2}$-code $\C$ of length $n$ with $\gcd(n,q)=1$. There are no further restrictions on $\C$.
\begin{theorem} \label{main cyclic}
    Let $w(X), \ell(X), f(X),g(X)\in \F_q[X]$ such that 
    \[X^n-1= w(X) \ell(X) f(X)g(X),\] 
   where $\gcd(n,q)=1$.  
   Consider the general $\F_q$-linear cyclic $\F_{q^2}$-code \[\C= \langle 
 w(X) f(X)+ \gamma w(X)q(X), \gamma w(X)g(X)\rangle .\]  
 Then 
\begin{align}
     \C^{\perp_{TE}} = & \Big\{ c(X)\ell^*(X)+\gamma d(X)\ell^*(X)f^*(X)\in \F_{q^2}[X] \mid c(X), d(X) \in \F_{q}[X] \mbox{ with } \label{eq:2.12b}\\ \nonumber
     & \qquad X^{n-\deg f(X)} c(X)+ \gamma^2 \widehat{q}(X)d(X)\in \langle g^*(X) \rangle \Big\}\\
     = & \bigg\langle  g^*(X) c'(X) \ell^*(X)+ \gamma g^*(X)d' (X)\ell^*(X)f^*(X),  \label{eq:2.12c}\\ \nonumber
     & \qquad\qquad\qquad  \frac{\gamma^2 \widehat{q}(X)}{h(X)}\ell^*(X)  - \gamma  \frac{X^{n-\deg f(X)}}{h(X)} \ell^*(X) f^*(X) \bigg\rangle ,
 \end{align}
 where $h(X)= \gcd \left( X^{n-\deg f(X)}, \widehat{q}(X) \right)$ and $c'(X), d'(X)\in \F_q[X]$ are such that 
\[  X^{n-\deg f(X)}c'(X) + \gamma^2 \widehat{q}(X) d'(X) = h(X).\]
 \end{theorem}

\begin{proof}
Let $ u(X)=a(X)\left(w(X) f(X)+\gamma w(X) q(X)\right)+b(X) \gamma w(X) g(X) \in \C $ and $v(X)=c(X)+\gamma d(X)\in \C^{\perp_{TE}}$  with $a(X),b(X),c(X),d(X)\in \F_q[X]$. Then
\begin{align*}
    0=&u(X)   \circledast v(X)\\
    =& a(X) w(X) f (X)  \circledast c(X)+a(X) w (X) f (X)  \circledast \gamma d(X)\\
    & +a(X) \gamma w(X) q (X)  \circledast c(X)+a(X) \gamma w(X) q (X)  \circledast \gamma d(X) \\
    & +b(X) \gamma w(X) g (X)  \circledast c(X)+b(X) \gamma w(X) g (X)  \circledast \gamma d(X) \\
    =& 2 a(X) w(X) f (X)  \star c(X)+\tr{(\gamma)}[a(X) w (X) f (X)  \star  d(X)]\\
    & +\tr{(\gamma)}[a(X)  w(X) q (X)  \star c(X)]+\tr{(\gamma^2)}[a(X) w(X) q (X)  \star  d(X)] \\
    & +\tr{(\gamma)}[b(X) w(X) g (X)  \star c(X)]+\tr{(\gamma^2)}[b(X)  w(X) g (X)  \star d(X)] \\
    = & 2 a(X) w(X) f (X) \star c(X) + 2 \gamma^2 a(X) w(X) q (X) \star d(X) \\
& + 2 \gamma^2 b(X) w(X) g (X) \star d(X).
\end{align*}

In particular, if $a(X)=0$ it follows that $2\gamma^2 b(X)w(X)g(X)\star d(X)=0.$ Thus, $d(X)\in \langle \ell^*(X) f^*(X)\rangle$ and we can write $d(X)= d_1(X)\ell^*(X)f^*(X)$ for some $d_1(X)\in \F_q[X].$ 

Then $2 \gamma^2 b(X) w(X) g (X) \star d_1(X)\ell^*(X)f^*(X)=0$ for all $b(X)\in \F_q[X]$ and
\begin{align*}
    0=& u(X)   \circledast v(X) \\
    = & 2 a(X) w(X) f (X) \star c(X) + 2 \gamma^2 a(X) w(X) q (X) \star d_1(X)\ell^*(X)f^*(X) \\
= &  2 a(X) w(X)  \star \left(  \widehat{f}(X)c(X)+ \gamma^2 \widehat{q}(X)d_1(X)\ell^*(X)f^*(X)  \right), \mbox{ by \eqref{lem:passtheq}}.
\end{align*}

Since all polynomials involved belong to $\F_q[X]$ and $\gamma^2 \in \F_q$, we have the necessary and sufficient condition  
\begin{align} \label{eq:2.12a}
\left(  \widehat{f}(X)c(X)+ \gamma^2 \widehat{q}(X)d_1(X)\ell^*(X)f^*(X)  \right) \in \langle \ell^*(X) f^*(X) g^*(X) \rangle.
\end{align}

Notice that
\begin{align*}
    \ell^*(X)  & f^*(X) g^*(X)\mid \left( \widehat{f}(X) c(X) +\gamma^2 \widehat{q}(X) d_1(X) \ell^*(X) f^*(X) \right) \\
    & \iff \ell^*(X) f^*(X) g^*(X) \mid \left( X^{n-\deg f(X)} f^*(X) c(X) +\gamma^2 \widehat{q}(X) d_1(X) \ell^*(X) f^*(X) \right) \\
    & \iff \ell^*(X) g^*(X) \mid \left( X^{n-\deg f(X)}  c(X) +\gamma^2 \widehat{q}(X) d_1(X) \ell^*(X)  \right) \\
    & \implies \ell^*(X)  \mid \left( X^{n-\deg f(X)}  c(X) +\gamma^2 \widehat{q}(X) d_1(X) \ell^*(X)  \right) \\
    & \implies \ell^*(X)  \mid \left( X^{n-\deg f(X)}  c(X) \right) \\
    & \implies \ell^*(X)  \mid   c(X) .
\end{align*}

Hence, any $c(X)+\gamma d(X)=c(X)+\gamma d_1(X) \ell^*(X) f^*(X) \in \C^{\perp_{TE}}$ has $c(X)$ of the form $c(X)= c_1(X) \ell^*(X)$ with $c_1(X)\in \F_q[X]$.

Substituting this back into \eqref{eq:2.12a} and the following statements, we have
\begin{align*}
    \ell^*(X) f^*(X) g^*(X) &\mid \left( \widehat{f}(X) c_1(X) \ell^*(X)  +\gamma^2 \widehat{q}(X) d_1(X) \ell^*(X) f^*(X) \right) \\
    & \iff \ell^*(X) g^*(X) \mid \left( X^{n-\deg f(X)}  c_1(X) \ell^*(X) +\gamma^2 \widehat{q}(X) d_1(X) \ell^*(X)  \right) \\
    & \iff  g^*(X) \mid \left( X^{n-\deg f(X)}  c_1(X) +\gamma^2 \widehat{q}(X) d_1(X)   \right) .
\end{align*}

\noindent Therefore,
\begin{align*}
     \C^{\perp_{TE}} = & \Big\{ c_1(X)\ell^*(X)+\gamma d_1(X)\ell^*(X)f^*(X)\in \F_{q^2}[X] \mid c_1(X), d_1(X) \in \F_{q}[X] \mbox{ such that }\\
     & \qquad X^{n-\deg f(X)} c_1(X)+ \gamma^2 \widehat{q}(X)d_1(X)\in \langle g^*(X) \rangle \Big\}.
 \end{align*}

Let $h(X):= \gcd \left( X^{n-\deg f(X)},  \widehat{q}(X) \right)$. Perform the extended Euclidean algorithm to find $c'(X), d'(X)\in \F_q[X]$ such that
\[  X^{n-\deg f(X)}c'(X) + \gamma^2 \widehat{q}(X) d'(X) = h(X).\]

Since $g^*(X)$ has nonzero constant term, we have that $\gcd\left(g^*(X), h(X)  \right)=1$. Hence,
\begin{align*}
    g^*(X) \mid &\left( X^{n-\deg f(X)}  c_1(X) +\gamma^2 \widehat{q}(X) d_1 (X)  \right) \\
    &\iff g^*(X) h(X) \mid \left( X^{n-\deg f(X)}  c_1(X) +\gamma^2 \widehat{q}(X) d_1(X)   \right) . 
\end{align*}

Notice that 
\[  X^{n-\deg f(X)}\left( g^*(X) c'(X)\right) + \gamma^2 \widehat{q}(X) \left( g^*(X) d'(X) \right) = g^*(X)h(X).\]

Therefore $\left(g^*(X) c'(X) \ell^*(X)+ \gamma g^*(X) d'(X) \ell^*(X)f^*(X)\right)\in \C^{\perp_{TE}}$ as a generator of the (inhomogeneous) solutions to $(0\ne) X^{n-\deg f(X)} c_1(X)+ \gamma^2 \widehat{q}(X)d_1(X)\in \langle g^*(X) \rangle$.

The solutions of the homogeneous equation
\[  X^{n-\deg f(X)} c_1(X) + \gamma^2 \widehat{q}(X) d_1(X)  = 0\]
%
are of the form
\[  c_1(X)=\frac{\gamma^2 \widehat{q}(X)}{h(X)} e(X), d_1(X)=-\frac{X^{n-\deg f(X)}}{h(X)}e(X) \mbox{ with }e(X)\in \F_q[X]. \]

We conclude
\begin{align*}
     \C^{\perp_{TE}} = & \bigg\langle  g^*(X) c'(X) \ell^*(X)+ \gamma g^*(X)d' (X)\ell^*(X)f^*(X), \\
     & \qquad\qquad\qquad  \frac{\gamma^2 \widehat{q}(X)}{h(X)}\ell^*(X)  - \gamma  \frac{X^{n-\deg f(X)}}{h(X)} \ell^*(X) f^*(X)  \bigg\rangle . \qedhere
 \end{align*}
%
\end{proof}

\begin{example}
Let $\C$ be the code $\langle w(X) f(X), \gamma w(X)g(X)\rangle$, that is, choose $q(X)=0$. Then $\widehat{q}(X)=0$, $h(X)=\gcd(X^{n-\deg f(X)}, 0)=X^{n-\deg f(X)}$, and 
$X^{n-\deg f(X)} c'(X) + \gamma^2 \widehat{q}(X)d'(X)=h(X)$ 
for the constant polynomials $c'(X)=1$ and $d'(X)=0$. Then
\[ g^*(X) c'(X) \ell^*(X)+ \gamma g^*(X)d' (X)\ell^*(X)f^*(X)=  g^*(X)\ell^*(X)\in \C^{\perp_{TE}}, \]
\[ \frac{\gamma^2 \widehat{q}(X)}{h(X)}\ell^*(X)  - \gamma  \frac{X^{n-\deg f(X)}}{h(X)} \ell^*(X) f^*(X) = -\gamma  \ell^*(X) f^*(X)\in \C^{\perp_{TE}},  \]
\[\mbox{ and } \qquad \C^{\perp_{TE}} =\left\langle g^*(X)\ell^*(X),-\gamma  \ell^*(X) f^*(X) \right\rangle, \]
which matches Theorem \ref{verma cor}.
\end{example}

\begin{example}
    Let $\C = \langle w(X) f(X)+ \gamma w(X)q(X), \gamma w(X)g(X)\rangle$ with $X^n-1= w(X)\ell(X)f(X) g(X)$, $q(X)=1$ constant, and $g(X)=X-1$. Then $g^*(X)=1-X$ and, recalling the implicit description of $\C^{\perp_{TE}}$ in \eqref{eq:2.12b}, for $c(X), d(X) \in \F_q[X]$ we have that
\begin{align} \label{eq:2.14}
    g^*(X) \mid \left( X^{n-\deg f(X)}  c(X) +\gamma^2 \widehat{q}(X) d(X)   \right) & \iff 
    (1-X) \mid \left( X^{n-\deg f(X)}  c(X) +\gamma^2 d(X)   \right) \nonumber\\
    & \iff  1^{n-\deg f(X)}  c(1) +\gamma^2 d(1)   =0 \nonumber\\
    & \iff    c(1) +\gamma^2 d(1)   =0.
\end{align}    

On the other hand, following the explicit construction of $\C^{\perp_{TE}}$ in \eqref{eq:2.12c}, gives \[h(X)=\operatorname{gcd}(X^{n-\deg f(X)}, \widehat{q}(X)) = \operatorname{gcd}(X^{n-\deg f(X)}, 1 ) =1,\] and we can write 
$X^{n-\deg f(X)}c'(X) +\gamma^2 \widehat{q}(X) d'(X)=h(X)$
with constant polynomials $c'(X)=0$ and $d'(X)= \gamma^{-2}$. Then 
\begin{align*}
    g^*(X)c'(X)\ell^*(X)+ \gamma g^*(X) d'(X)\ell^*(X) f^*(X)&= \gamma g^*(X) \gamma^{-2}\ell^*(X) f^*(X)\\
    &= \gamma^{-1} (1-X) \ell^*(X) f^*(X) \in \C^{\perp_{TE}},\\
%
%
%
%
    \frac{\gamma^2 \widehat{q}(X)}{h(X)}\ell^*(X)  - \gamma  \frac{X^{n-\deg f(X)}}{h(X)} \ell^*(X) f^*(X) &=
    \gamma^2 \ell^*(X)- \gamma X^{n-\deg f(X)} \ell^*(X)f^*(X)\\ &= \ell^*(X)\left( \gamma^2 - \gamma X^{n-\deg f(X)}f^*(X)  \right)\in \C^{\perp_{TE}},
\end{align*}
\[\mbox{ and } \qquad  \C^{\perp_{TE}}= \left\langle  \gamma^{-1} (1-X) \ell^*(X) f^*(X), \ell^*(X)\left( \gamma^2 - \gamma X^{n-\deg f(X)}f^*(X)  \right)  \right\rangle . \]

%

Notice that $\gamma^{-1} (1-X) \ell^*(X) f^*(X)=c(X)\ell^*(X)+ \gamma d(X)\ell^*(X)f^*(X)$ with $c(X)=0$, $d(X)=\gamma^{-2}(X-1)$. In particular, $c(1)+ \gamma^2 d(1)= 0 + \gamma^2 \gamma^{-2}(1-1)=0$,
which verifies Equation \eqref{eq:2.14}.

Similarly, $\ell^*(X)( \gamma^2 - \gamma X^{n-\deg f(X)}f^*(X)  )=c(X)\ell^*(X)+ \gamma d(X)\ell^*(X)f^*(X)$ with $c(X)=\gamma^2$, $d(X)=-X^{n-\deg f(X)}$. In particular, $c(1)+ \gamma^2 d(1)= \gamma^2 -1^{n-\deg f(X)} \gamma^2 =0$,
again verifying Equation \eqref{eq:2.14}.
%
%

\end{example}

\begin{example}
    Let $\C$ be the $\F_3$-linear cyclic $\F_9$-code $\langle w(X) f(X)+ \gamma w(x)q(X), \gamma w(X)g(X)\rangle$ with $n=28$, $w(X)=X+2$, $f(X)=X^2+1$, $g(X)=X^6+2X^5+2X^3+2X+1$, and $q(X)=X$ as in \cite[Table 1]{Verma-Sharma:2024}. Then
    \begin{align*}
       \ell(X)&=  X^{19}+2X^{18}+2X^{17}+X^{14}+2X^{13}+2X^{12}+2X^{11}+2X^{10}\\
        & \qquad\qquad  +2X^{9}+2X^{8}+2X^{7}+2X^{6}+X^{5}+2X^{2}+2X+1,\\
        w^*(X)=2X+1, &\quad \ell^*(X)=\ell(X), \quad f^*(x)=f(X), \quad g^*(X)=g(X), \quad \widehat{q}(X)=X^{27}.
    \end{align*}
    Then $h(X)=\gcd( X^{n-\deg f(X)},  \widehat{q}(X))=\gcd( X^{26},  X^{27})=X^{26}$, and $c'(X)=1$, $d'(X)=0$ solve the equation
    \[ X^{n-\deg f(X)}c'(X) +\gamma^2 \widehat{q}(X)d'(X)= h(X). \]
    Hence,
    \begin{align*}
    g^*(X) &c'(X) \ell^*(X)+ \gamma g^*(X)d' (X)\ell^*(X)f^*(X)=g(X)\ell(X)\\
    &= X^{25}+X^{24}+X^{21}+X^{20}+X^{17}+X^{16}+X^{13}\\
    &+ X^{12}+X^9+X^8+X^5+X^4+X+1
    \end{align*}
    gives one generator of $\C^{\perp_{TE}}$. Note that $\gamma^2\in \F_3$ leaves the only option $\gamma^2=-1$, and another generator of $\C^{\perp_{TE}}$ thus is given by
    \begin{align*}
        & \frac{\gamma^2 \widehat{q}(X)}{h(X)}\ell^*(X)  - \gamma  \frac{X^{n-\deg f(X)}}{h(X)} \ell^*(X) f^*(X) = \dfrac{\gamma^2 X^{27}}{X^{26}}\ell(X)- \gamma\dfrac{ X^{26}}{X^{26}}\ell(X)f(X)\\
        = & -X \ell(X) - \gamma \ell(X)f(X)\\ 
        = & -X^{20}+X^{19}+X^{18}-X^{15}+X^{14}+X^{13}+X^{12}+ X^{11}+X^{10}\\
          & +X^{9}+X^{8}+X^{7}-X^{6}+X^{3}+X^2-X\\
          & +\gamma (-X^{21}+X^{20}+X^{18}+X^{17}-X^{16}+X^{15}-X^{13}-X^{12}-X^{11}\\
          & -X^{10}-X^9-X^8+X^6-X^5+X^4+X^3+X-1).
    \end{align*}
    
    Thus, we have the trace Euclidean dual
    \begin{align*}
    \C^{\perp_{TE}} = &\left\langle g(X)\ell(X) , - X \ell(X) - \gamma \ell(X)f(X) \right\rangle \\
    = \big\langle X^{25}&+X^{24}+X^{21}+X^{20}+X^{17}+X^{16}+X^{13}\\
    &+X^{12}+X^9+X^8+X^5+X^4+X+1,\\
    &-X^{20}+X^{19}+X^{18}-X^{15}+X^{14}+X^{13}+X^{12} +X^{11}\\
    &+X^{10}+X^{9}+X^{8}+X^{7}-X^{6}+X^{3}+X^2-X\\
    &+\gamma (-X^{21}+X^{20}+X^{18}+X^{17}-X^{16}+X^{15}-X^{13}- X^{12}\\
    &-X^{11}-X^{10}-X^9-X^8+X^6-X^5+X^4+X^3+X-1)
    \big\rangle.
    \end{align*}
\end{example}

Actually, the restriction $\gcd(n,q)=1$ can easily be dropped from Theorem \ref{main cyclic}. In particular, we are including the following two results that can be derived by calculations almost identical to the ones performed for Theorem \ref{main skew cyclic}.

\begin{theorem} 
    Let $w(X), \ell(X), f(X),g(X)\in \F_q[X]$ such that 
    \[X^n-1= w(X) \ell(X) f(X)g(X).\]  
   Consider the general $\F_q$-linear cyclic $\F_{q^2}$-code \[\C= \langle 
 w(X) f(X)+ \gamma q(X), \gamma w(X)g(X)\rangle .\]  
 Then 
\begin{align*}
     \C^{\perp_{TE}} = \bigg\langle c'(X) \dfrac{ w^*(X)\ell^*(X)g^*(X) }{k'(X)} &+ \gamma d'(X) \dfrac{w^*(X)\ell^*(X)g^*(X)}{k'(X)}\ell^*(X)f^*(X), \\ \nonumber
      \dfrac{\gamma^2  \ell^*(X)  \widehat{q}(X)}{h'(X)} &- \gamma \dfrac{X^{n-\deg f(X)} \widehat{w}(X)}{h'(X)} \ell^*(X) f^*(X) \bigg\rangle, 
 \end{align*}
where 
\begin{align*}
&h'(X) = \gcd( X^{n-\deg f(X)}\widehat{w}(X)), \ell^*(X) \widehat{q}(X)),\\ &k'(X) = \gcd( h'(X), w^*(X)\ell^*(X)g^*(X) ).
\end{align*}
and 
\[c'(X), d'(X) \in \F_q[X]\] 
are such that
\[c'(X) X^{n-\deg f(X)}\widehat{w}(X) +\gamma^2 d'(X) \ell^*(X) \widehat{q}(X) = h'(X).  \]
 \end{theorem}

\begin{proof}
Follow the proof of Theorem \ref{main skew cyclic} and replace any instance of $-X$ in this proof with $X$. 
Note that $h(X) := \gcd (\widehat{w}(X)\widehat{f}(X), \ell^*(X)f^*(X)\widehat{q}(X))= f^*(X)h'(X)$ with $h'(X):= \gcd (X^{n-\deg f(X)}\widehat{w}(X), \ell^*(X)\widehat{q}(X))$.
\end{proof}

We also include the matching result on the trace Hermitian dual.

 \begin{theorem} 
    Let $w(X), \ell(X), f(X),g(X)\in \F_q[X]$ such that 
    \[X^n-1= w(X) \ell(X) f(X)g(X).\]  
   Consider the general $\F_q$-linear cyclic $\F_{q^2}$-code \[\C= \langle 
 w(X) f(X)+ \gamma q(X), \gamma w(X)g(X)\rangle .\]  
 Then 
\begin{align*}
     \C^{\perp_{TH}} = \bigg\langle c'(X) \dfrac{ w^*(X)\ell^*(X)g^*(X) }{k'(X)} &+ \gamma d'(X) \dfrac{w^*(X)\ell^*(X)g^*(X)}{k'(X)}\ell^*(X)f^*(X), \\ \dfrac{\gamma^{q+1}  \ell^*(X)  \widehat{q}(X)}{h'(X)} &- \gamma \dfrac{X^{n-\deg f(X)}\widehat{w}(X)}{h'(X)} \ell^*(X) f^*(X) \bigg\rangle, 
 \end{align*}
where 
\begin{align*} 
h'(X) & = \gcd(X^{n-\deg f(X)} \widehat{w}(X), \ell^*(X) \widehat{q}(X)),\\
k'(X) & = \gcd( h'(X), w^*(X)\ell^*(X)g^*(X) ),
\end{align*}
and \[c'(X), d'(X) \in \F_q[X]\] are such that

\[c'(X) X^{n-\deg f(X)} \widehat{w}(X)  +\gamma^{q+1} d'(X) \ell^*(X) \widehat{q}(X) = h'(X).\]
Recall here that $\gamma^{q+1}=-\gamma^2$, see Lemma \ref{gammasquare}.
 \end{theorem}




\section{\texorpdfstring{The Trace Dual for $\F_{q}$-Linear Skew Cyclic $\F_{q^{2}}$-Codes}{The Trace Dual for Fq-Linear Skew Cyclic Fq2-Codes}}\label{sec:dualskewcyclicodes}

Consider a field extension $\F_q\leq \F_{q^2}$ of degree $2$, where $q$ is an odd prime power. Let $\operatorname{Tr}(x)= x+ x^q$ for $x\in \F_{q^2}$ denote the trace and $\sigma: \F_{q^2}\rightarrow \F_{q^2}$ the Frobenius automorphism given by $\sigma(x)=x^q$.
Further, let $\F_{q^2}[X; \sigma]$ and $\F_{q}[X; \sigma]=\F_{q}[X]$ be the rings of skew polynomials with coefficients in $\F_{q^2}$ and $\F_q$, respectively.

There exists $\gamma \in \F_{q^2}\setminus \F_q$ such that $\tr(\gamma)=0$, that is $\sigma(\gamma) =\gamma^q = -\gamma$. Then, by Lemma \ref{gammasquare}, $\gamma^2\in \F_q$ and $\tr(\gamma^2)=2\gamma^2 \ne 0$.
%
%
Moreover, $\ker{(\tr)}= \gamma \F_q$, and $\{ 1,\gamma \}$ is an $\F_q$-basis of $\F_{q^2}$. We can write:
\[\F_{q^2} = \F_q + \gamma \F_q \] and \[\F_{q^2}[X;\sigma] = \F_q [X;\sigma]+ \gamma \F_q[X;\sigma] = \F_q [X]+ \gamma \F_q[X].\]

We also note the following useful multiplication fact on the skew polynomial ring $\F_{q^2}[X;\sigma]$.

\begin{proposition}
For all $f(X) \in \F_{q}[X] \subseteq \F_{q^2}[X;\sigma]$ holds $f(X)\gamma =\gamma f(-X)$.
\end{proposition}

\begin{proof}
This is obvious from $X\gamma =\sigma(\gamma)X=\gamma^qX=-\gamma X$ in $\F_{q^2}[X;\sigma]$.
\end{proof}

\begin{definition}
    An \emph{$\F_q$-linear skew cyclic $\F_{q^2}$-code of length $n$} is an $\F_q$-linear subspace $\C$ of $\F_{q^2}^n$ that satisfies
\[(c_0, c_1, \ldots, c_{n-1}) \in \C \implies (\sigma(c_{n-1}),\sigma(c_0),\ldots,\sigma(c_{n-2}))\in \C.\]
This definition requires $\sigma^n=\operatorname{id}$. Thus $n$ must be even.
\end{definition}

In order to study dual codes of $\F_q$-linear skew cyclic $\F_{q^2}$-codes of length $n$, we will make use of skew polynomials with coefficients in $\F_{q^2}$ and some associated polynomial algebras. Denote
\[R_n := \F_{q^2}[X; \sigma]/\left( X^n-1\right),\qquad\ \]
\[ P_n:=\F_{q}[X;\sigma]/\left( X^n-1\right)\subseteq R_n.\]
Any element of $R_n$ can uniquely be written as $f(X)+(X^n-1)$ for a suitable representative $f(X)=\sum_{i=0}^{n-1}f_i X^i\in \F_{q^2}[X;\sigma]$, and we will often conflate $f(X)+(X^n-1)\in R_n$ with its representative $f(X)=\sum_{i=0}^{n-1}f_i X^i\in \F_{q^2}[X;\sigma]$ for ease of notation. Notice that $\F_{q}[X;\sigma]=\F_{q}[X]$ and $P_n=\F_{q}[X]/\left( X^n-1\right)$ are commutative under multiplication while $R_n$ is not commutative. Further, note that
$\C\subseteq R_n$ is an $\F_q$-linear skew cyclic $\F_{q^2}$-code of length $n$ if and only if  $\C$ is an $\F_{q}[X]$-submodule of $R_n$.

Given $f(X)+(X^n-1)$, $g(X)+(X^n-1)\in R_n$ with $f(X)=\sum_{i=0}^{n-1}f_i X^i$ and $g(X)=\sum_{i=0}^{n-1}g_iX^i$, we again define the \emph{Euclidean inner product} as
$f(X) \star g(X):= \langle (f_0, f_1,\ldots, f_{n-1}), (g_0, g_1, \ldots, g_{n-1})  \rangle$,
the \emph{Hermitian inner product} as
$f(X) \bullet g(X):= \langle (f_0, f_1,\ldots, f_{n-1}), (g_0, g_1, \ldots, g_{n-1})  \rangle_H$,
the \emph{trace Euclidean inner product} as
$f(X) \circledast  g(X):= \tr \left( f(X) \star g(X) \right)$,
and the \emph{trace Hermitian inner product} as
$f(X)  \boxdot g(X):= \tr \left( f(X) \bullet g(X) \right)$.
Proposition \ref{prop 2.6} remains valid.
%
   %

\cite[Theorem 3.2]{Verma-Sharma:2024} and \cite[Lemma 3.4]{Verma-Sharma:2024} easily adapt to the case of skew cyclic codes.

\begin{proposition}\label{codeshape1}
    Let $\C$ be an $\F_q$-linear skew cyclic $\F_{q^2}$-code of length $n$. There exist $f(X),g(X),q(X)\in \F_q[X]$ such that $f(X)$ and $g(X)$ divide $ X^n-1$ and $\C= \langle f(X)+ \gamma q(X), \gamma g(X) \rangle$.
\end{proposition}

\begin{proof}
    Let $\psi: \C \rightarrow P_n$ be given by 
    \[ \psi \left( \sum_{i=0}^{n-1}(a_i +\gamma b_i)X^i \right) = \sum_{i=0}^{n-1}a_iX^i \]
    for all $a_i,b_i \in \F_q$. Note that $\psi$ is an $\F_q[X]$-module homomorphism and $\operatorname{Im}{\psi}$ is an $\F_q[X]$-submodule of $P_n$. Therefore, $\operatorname{Im}{\psi}$ is a linear cyclic $\F_q$-code of length $n$ and $\operatorname{Im}{\psi}= \langle f(X)\rangle$ for some $f(X)\in \F_q[X]$ with $f(X)\mid X^n-1$. 

    Besides, if $M=\left\{d(X)\in \F_q[X] : \gamma d(X)\in \ker \psi\right\}$, that is, $\ker\psi = \gamma M$, we have that $M$ is an $\F_q[X]$-submodule of $P_n$. Therefore, $M$ is a linear cyclic $\F_q$-code of length $n$ and $M= \langle g(X)\rangle$ for some $g(X)\in \F_q[X]$ with $g(X)\mid X^n-1$. Hence,
    $\C  = \langle f(X)+\gamma q(X),\gamma g(X)\rangle$
    where $q(X)\in \F_q[X]$ is chosen such that $f(X)+ \gamma q(X)\in \C$.
\end{proof}

\begin{proposition}\label{codeshape2}
    Let $f(X),g(X),q(X)$ as in Proposition \ref{codeshape1}. Then $q(X)$ is uniquely determined by $f(X)$ and $g(X)$ by requiring $\deg q(X)< \deg g(X)$.
\end{proposition}

\begin{proof}
    By performing the division algorithm, we have the uniqueness of $a(X),r(X)\in \F_q[X]$ with $q(-X)= a(X) g(-X)+ r(X)$ and $\deg r(X)< \deg g(X)$. Thus,
    \begin{align*}
    \langle f(X)+\gamma q(X),\gamma g(X)\rangle &= \langle f(X)+ q(-X)\gamma, g(-X)\gamma\rangle\\ &= \langle f(X)+ r(X)\gamma,g(-X)\gamma\rangle= \langle f(X)+\gamma r(-X),\gamma  g(X)\rangle. \qedhere
    \end{align*}
\end{proof}

\begin{lemma}
    If $f(X),g(X),q(X)$ are as in Propositions \ref{codeshape1} and \ref{codeshape2}, then $g(X)\mid \left(\frac{X^n-1}{f(-X)}  q(X)\right)$.
\end{lemma}

\begin{proof}
    Let $c(X)= \frac{X^n-1}{f(X)}\left(f(X)+\gamma q(X)\right)\in \C$. Recall that $n$ is even and notice that
    \[\psi(c(X))= \psi \left( \frac{X^n-1}{f(X)} \gamma q(X)  \right) = \psi \left(\gamma \frac{(-X)^n-1}{f(-X)}  q(X)  \right)= \psi \left(\gamma \frac{X^n-1}{f(-X)}  q(X)  \right)=0. \]
    Thus, $\gamma \frac{X^n-1}{f(-X)}  q(X) \in \ker \psi$,  $\frac{X^n-1}{f(-X)}  q(X)\in \langle g(X)\rangle$, and thus $g(X)\mid \left(\frac{X^n-1}{f(-X)}  q(X)\right)$.
\end{proof}

\begin{theorem}
    Let $\C$ be an $\F_q$-linear skew cyclic $\F_{q^2}$-code of length $n$. Then there exist $w(X)$, $\ell(X)$, $f(X)$, $g(X)$, $q(X)\in \F_q[X]$ such that $\deg(q)<n$,
    \[X^n-1= w(X)\ell(X)f(X) g(X)\quad \mbox{ and, }\quad \C = \langle w(X) f(X)+ \gamma q(X), \gamma w(X)g(X)\rangle. \]
    
\end{theorem}

\begin{proof}
    Recall that there exist $f(X),g(X),q(X)\in \F_q[X]$ such that $f(X)$ and $g(X)$ divide $ X^n-1$ and $\C= \langle f(X)+ \gamma q(X), \gamma g(X) \rangle$. Let $w(X)= \gcd\left(f(X),g(X)\right)$. Then we can write $f(X)=w(X)f_1(X)$ and $g(X)=w(X)g_1(X)$ for some $f_1(X),g_1(X)\in \F_q[X]$. Since $f(X),g(X)\mid X^n-1$, we have $\operatorname{lcm}\left(f(X),g(X)\right) =w(X)f_1(X)g_1(X) \mid X^n-1$ and
    \[ \C = \langle f(X)+ \gamma q(X), \gamma g(X) \rangle = \langle w(X)f_1(X)+ \gamma q(X), \gamma w(X)g_1(X) \rangle. \qedhere\]
\end{proof}

Next we discuss the trace Euclidean dual code of a special family of $\F_q$-linear skew cyclic $\F_{q^2}$-codes.
\begin{theorem}\label{verma skew cor}
    Let $w(X), \ell(X), f(X),g(X)\in \F_q[X]$ such that 
    \[X^n-1= w(X) \ell(X) f(X)g(X).\] 
    
    Define the $\F_q$-linear skew cyclic $\F_{q^2}$-codes \[\C= \langle 
 w(X) f(X), \gamma w(X)g(X) \rangle\quad \mbox{ and }\quad \mathcal{D}= \langle 
 \ell^*(X) g^*(X), \gamma \ell^*(X)f^* (X) \rangle.\]  Then $\C^{\perp_{TE}}=\mathcal{D}.$
\end{theorem}

\begin{proof}
    To show $\C^{\perp_{TE}}\subseteq \mathcal{D}$, take $c(X)+ \gamma d(X)\in \C^{\perp_{TE}}$  with $c(X),d(X)\in \F_q[X]$. Then, for any $a(X), b(X)\in \F_q[X]$ we have that 
    \[ \left(a(X)w(X)f(X)+b(X)\gamma w(X)g(X)\right)\circledast \left( c(X) + \gamma d(X) \right) =0.\]

    In particular, if $a(X)=0$, then 
    \begin{align*}
        0 &= b(X)\gamma w(X)g(X)\circledast  c(X) + b(X)\gamma w(X)g(X)\circledast \gamma d(X) \\&= \gamma b(-X)w(X)g(X)\circledast  c(X) + \gamma b(-X) w(X)g(X)\circledast \gamma d(X) \\
        &= \tr{(\gamma)}[b(-X) w(X)g(X)\star  c(X)] +\tr{(\gamma^2)}[b(-X) w(X)g(X)\star  d(X)]\\
        &= \tr{(\gamma^2)} [b(-X) w(X)g(X)\star  d(X)].
    \end{align*}
    Thus, $d(X)\in \langle \left[ \left(X^n-1\right)/\left(w(X) g(X)\right)\right]^* \rangle= \langle  \ell^*(X) f^*(X) \rangle.$

    Besides, if $b(X)=0$, then 
    \begin{align*}
        0 &= a(X)w(X)f(X)\circledast  c(X) + a(X)w(X)f(X)\circledast \gamma d(X) \\
        &= a(X)w(X)f(X)\circledast  c(X) +\tr{(\gamma)} [a(X)w(X)f(X)\circledast d(X)] \\
        &= a(X)w(X)f(X)\circledast  c(X).
    \end{align*}
    Thus, $c(X)\in \langle \left[ \left(X^n-1\right)/\left(w(X) f(X)\right)\right]^* \rangle= \langle  \ell^*(X) g^*(X) \rangle$, and $\C^{\perp_{TE}}\subseteq \mathcal{D}$ follows. From the above calculations, it is now easy to also verify that $\mathcal{D}\subseteq \C^{\perp_{TE}}$. 
\end{proof}

\begin{theorem} \label{main skew cyclic}
    Let $w(X), \ell(X), f(X),g(X)\in \F_q[X]$ such that 
    \[X^n-1= w(X) \ell(X) f(X)g(X).\]  
   Consider the general $\F_q$-linear skew cyclic $\F_{q^2}$-code \[\C= \langle 
 w(X) f(X)+ \gamma q(X), \gamma w(X)g(X)\rangle .\]  
 Then 
\begin{align*}
     \C^{\perp_{TE}} = & \bigg\langle c'(X) \dfrac{X^n-1}{k(X)} + \gamma d'(-X) \dfrac{X^n-1}{k(-X)}\ell^*(X)f^*(X), \\ \nonumber
      & \qquad\qquad\qquad  \dfrac{\gamma^2  \ell^*(-X) f^*(-X) \widehat{q}(-X)}{h(X)}
     - \gamma \dfrac{\widehat{w}(-X)\widehat{f}(-X)}{h(-X)} \ell^*(X) f^*(X) \bigg\rangle, 
 \end{align*}
%
where $h(X)= \gcd( \widehat{w}(X) \widehat{f}(X), \ell^*(-X)f^*(-X) \widehat{q}(-X))$, $k(X)=\gcd( h(X), X^n-1 )$, and $c'(X), d'(X)\in \F_q[X]$ are such that
\[   c'(X) \widehat{w}(X) \widehat{f}(X) +\gamma^2 d'(X) \ell^*(-X) f^*(-X)\widehat{q}(-X) = h(X).  \]
 \end{theorem}

\begin{proof}
Let $ u(X)=a(X)\left(w(X) f(X)+\gamma  q(X)\right)+b(X) \gamma w(X) g(X) \in \C $ and $v(X)=c(X)+\gamma d(X)\in \C^{\perp_{TE}}$ with $a(X),b(X),c(X),d(X)\in \F_q[X]$. Then
\begin{align*}
    0=& u(X)   \circledast v(X) \\
    =& a(X) w(X) f (X)  \circledast c(X)+a(X) w (X) f (X)  \circledast \gamma d(X)\\
    & +\gamma a(-X)   q (X)  \circledast c(X)+\gamma a(-X)   q (X)  \circledast \gamma d(X) \\
    & +\gamma b(-X) w(X) g (X)  \circledast c(X)+\gamma b(-X)  w(X) g (X)  \circledast \gamma d(X) \\
    =& 2 a(X) w(X) f (X)  \star c(X)+\tr{(\gamma)}[a(X) w (X) f (X)  \star  d(X)]\\
    & +\tr{(\gamma)}[a(-X)   q (X)  \star c(X)]+\tr{(\gamma^2)}[a(-X)  q (X)  \star  d(X)] \\
    & +\tr{(\gamma)}[b(-X) w(X) g (X)  \star c(X)]+\tr{(\gamma^2)}[b(-X)  w(X) g (X)  \star d(X)] \\
    = & 2 a(X) w(X) f (X) \star c(X) + 2 \gamma^2 a(-X)  q (X) \star d(X) \\
& + 2 \gamma^2 b(-X) w(X) g (X) \star d(X).
\end{align*}

In particular, if $a(X)=0$ it follows that $2\gamma^2 b(-X)w(X)g(X)\star d(X)=0$ for all $b(X)\in \F_q[X]$. Thus, $d(X)\in \langle \ell^*(X) f^*(X)\rangle$ and we can write $d(X)= d_1(X)\ell^*(X)f^*(X)$ for some $d_1(X)\in \F_q[X].$ 

Then, $2 \gamma^2 b(-X) w(X) g (X) \star d_1(X)\ell^*(X)f^*(X)=0$ for all $b(X)\in \F_q[X]$ and
\begin{align*}
    0&= u(X)   \circledast v(X) \\
    \iff 0 &=  2 a(X) w(X) f (X) \star c(X) + 2 \gamma^2 a(-X) q (X) \star d_1(X)\ell^*(X)f^*(X) \\
    \iff 0 &= a(X)w(X) f(X) \star c(X) + \gamma^2 a(-X) q(X) \star d_1(X) \ell^*(X) f^*(X) \\
    \iff 0 &= a(X)w(X) f(X) \star c(X) + \gamma^2 a(X) q(-X) \star d_1(-X) \ell^*(-X) f^*(-X), \mbox{ by }\eqref{lem:flipsings} \\
    \iff 0 &= a(X) \star \left( c(X) \widehat{w}(X) \widehat{f}(X) +  \gamma^2 d_1(-X) \ell^*(-X) f^*(-X) \widehat{q}(-X)\right)
\end{align*}
with $\widehat{q(-X)}=\widehat{q}(-X)$. From this equality, we can see
\[c(X) \widehat{w}(X) \widehat{f}(X) +  \gamma^2 d_1(-X) \ell^*(-X) f^*(-X)\widehat{q}(-X) \in ( X^n-1 ), \] and thus we have the necessary and sufficient condition
\[X^n-1 \mid \left( c(X) \widehat{w}(X) \widehat{f}(X) +\gamma^2 d_1(-X) \ell^*(-X) f^*(-X) \widehat{q}(-X)\right)\]

Let $h(X)= \gcd(\widehat{w}(X) \widehat{f}(X), \ell^*(-X)f^*(-X)\widehat{q}(-X)  )$ and $k(X)=\gcd( h(X), X^n-1 )$. Perform the extended Euclidean algorithm to find $c'(X), d'(X)\in \F_q[X]$ such that
\[  c'(X) \widehat{w}(X) \widehat{f}(X) +\gamma^2 d'(X) \ell^*(-X) f^*(-X)\widehat{q}(-X) = h(X).  \]
Then
\begin{align*}
   \operatorname{lcm}(h(X),X^n-1)=& c'(X) \frac{X^n-1}{k(X)} \widehat{w}(X) \widehat{f}(X) \\
   &+ \gamma^2 d'(X) \frac{X^n-1}{k(X)}  \ell^*(-X) f^*(-X)  \widehat{q}(-X).
\end{align*}
%
Therefore, 
\[ c'(X) \frac{X^n-1}{k(X)} + \gamma d'(-X) \frac{X^n-1}{k(-X)} \ell^*(X)f^*(X) \in \C^{\perp_{TE}}\]
as a generator of the (inhomogeneous) solutions to 
\[(0\ne)\ c(X) \widehat{w}(X) \widehat{f}(X)+\gamma^2 d_1(-X) \ell^*(-X) f^*(-X) \widehat{q}(-X) \in ( X^n-1 ).\] 

The solutions of the homogeneous equation
%
\[  c(X) \widehat{w}(X) \widehat{f}(X) +\gamma^2 d_1(-X) \ell^*(-X) f^*(-X)\widehat{q}(-X)  = 0  \]
are of the form
\begin{align*}
    c(X)=\dfrac{\gamma^2  \ell^*(-X) f^*(-X)\widehat{q}(-X)}{h(X)}e(X), d_1(-X)=  -\dfrac{\widehat{w}(X)\widehat{f}(X)}{h(X)}e(X) \mbox{ with }e(X)\in \F_q[X].
\end{align*}
%
%
Therefore, 
\begin{align*}
& \dfrac{\gamma^2  \ell^*(-X) f^*(-X)\widehat{q}(-X)}{h(X)}e(X) - \gamma \dfrac{\widehat{w}(-X)\widehat{f}(-X)}{h(-X)}e(-X) \ell^*(X)f^*(X)\\ = & e(X)\bigg[\dfrac{\gamma^2  \ell^*(-X) f^*(-X)\widehat{q}(-X)}{h(X)} - \gamma \dfrac{\widehat{w}(-X)\widehat{f}(-X)}{h(-X)} \ell^*(X)f^*(X)\bigg] \in \C^{\perp_{TE}}
\end{align*}
and we discover $\frac{\gamma^2  \ell^*(-X) f^*(-X)\widehat{q}(-X)}{h(X)} - \gamma \frac{\widehat{w}(-X)\widehat{f}(-X)}{h(-X)} \ell^*(X)f^*(X)$ as another generator of $\C^{\perp_{TE}}$.
We conclude
\begin{align*}
     \C^{\perp_{TE}} = & \bigg\langle c'(X) \dfrac{X^n-1}{k(X)} + \gamma d'(-X) \dfrac{X^n-1}{k(-X)}\ell^*(X)f^*(X), \\
     & \qquad\qquad\qquad  \dfrac{\gamma^2  \ell^*(-X) f^*(-X)\widehat{q}(-X)}{h(X)} 
     - \gamma \dfrac{\widehat{w}(-X)\widehat{f}(-X)}{h(-X)} \ell^*(X) f^*(X)   \bigg\rangle .\qedhere
 \end{align*}
\end{proof}
%
%
%
%

\begin{example}
Let $\C$ be the code $\langle w(X) f(X), \gamma w(X)g(X)\rangle$, that is, choose $q(X)=~0$. Then $\widehat{q}(X)=0$, 
\begin{align*}
    h(X)&= \gcd( \widehat{w}(X) \widehat{f}(X), 0 )= \widehat{w}(X) \widehat{f}(X)= X^{2n - \deg w(X) - \deg f(X)}w^*(X)f^*(X),\\
    k(X)&= \gcd( \widehat{w}(X) \widehat{f}(X), X^n-1  ) = w^*(X) f^*(X),
\end{align*}
and $c'(X) \widehat{w}(X) \widehat{f}(X)+\gamma^2 d'(X) \ell^*(-X) f^*(-X)\widehat{q}(-X) = h(X)$ for the constant polynomials $c'(X)=1$ and $d'(X)=0$. Then
%
%
%
%
%
\begin{align*}
     1 \cdot \dfrac{X^n-1}{k(X)} + \gamma \cdot 0 \cdot \dfrac{X^n-1}{k(-X)}\ell^*(X)f^*(X) &= \dfrac{X^n-1}{w^*(X) f^*(X)} = - \ell^*(X)g^*(X)\in \C^{\perp_{TE}} ,\\
    \dfrac{\gamma^2  \ell^*(-X) f^*(-X) \cdot 0}{h(X)} &- \gamma \dfrac{\widehat{w}(-X)\widehat{f}(-X)}{h(-X)}\ell^*(X) f^*(X)\\
    & =  -\gamma \dfrac{\widehat{w}(-X)\widehat{f}(-X) }{ \widehat{w}(-X)\widehat{f}(-X) }\ell^*(X) f^*(X)\\
    & = -\gamma  \ell^*(X) f^*(X)\in \C^{\perp_{TE}},
\end{align*}
and $ \C^{\perp_{TE}} = \left\langle  - \ell^*(X)g^*(X) , -\gamma  \ell^*(X) f^*(X) \right \rangle$, which matches Theorem \ref{verma skew cor}.
\end{example}

\begin{example}
    Let $\C$ be the $\F_3$-linear skew cyclic $\F_9$-code $\langle w(X) f(X)+ \gamma q(X), \gamma w(X)g(X)\rangle$ with $n=10$, $X^n-1= w(X)\ell(X)f(X) g(X)$, $w(X)= X+1$, $\ell(X) = X^4+X^3+X^2+X+1$, $f(X)=X^4-X^3+X^2-X+1$, $g(X)=X-1$, and $q(X)=w(X)$. Then
    \begin{align*}
        w^*(X) &=w(X),\qquad \widehat{w}(X)= X^9 w^*(X)= X^9 w(X),\qquad \ell^*(X)= \ell(X),\qquad \ell^*(-X) =f(X),\\
        f^*(X)&= f(X), \qquad f^*(-X)= \ell(X), \qquad \widehat{f}(X)= X^6 f^*(X)= X^6 f(X),\\
        \widehat{q}(X)&= \widehat{w}(X)= X^9 (X+1),\qquad \widehat{q}(-X)= X^9 (X-1)= X^9 g(X).
    \end{align*}
    Therefore, recalling $\gamma^2=-1$ for $\gamma \in \F_9$, we have
    \begin{align*} 
     h(X)&=\gcd ( \widehat{w}(X)\widehat{f}(X), \gamma^2 \ell^*(-X)f^*(-X) \widehat{q}(-X))\\
     &=\gcd (  X^9 w(X)X^6 f(X), - f(X)\ell(X)X^9 g(X))
     = X^9 f(X) \qquad \mbox{ and }\\
     k(X)&= \gcd \left(h(X),X^n-1\right)= \gcd \left(X^9 f(X),X^n-1\right)= f(X).
     \end{align*}
If we choose $c'(X)=2X^4+2X^3+X^2+2X+1$ and $d'(X)=2X^6 - X^5+X^4-X^3+X^2-X+1$, then
\begin{align*}
&c'(X) \widehat{w}(X) \widehat{f}(X) +\gamma^2 d'(X) \ell^*(-X) f^*(-X)\widehat{q}(-X)\\
= &c'(X) X^9 w(X)X^6 f(X) - d'(X)f(X)\ell(X)X^9 g(X)\\
= &X^9 f(X)=h(X).
\end{align*}

Notice that $k(-X)=f(-X)=\ell(X)$. Then
\begin{align*}
    c'(X) \dfrac{X^n-1}{k(X)} + \gamma & d'(-X) \dfrac{X^n-1}{k(-X)}\ell^*(X)f^*(X)\\
    & =  c'(X) w(X)\ell(X)g(X) + \gamma d'(-X) \dfrac{X^n-1}{\ell}\ell(X)f(X) \\
    & =  c'(X) w(X)\ell(X)g(X) = X^9 + 2X^5 + 2X^4 + 1 \in \C^{\perp_{TE}},
    \end{align*}

\begin{align*}
    \dfrac{\gamma^2  \ell^*(-X) f^*(-X) \widehat{q}(-X)}{h(X)}
     &- \gamma \dfrac{\widehat{w}(-X)\widehat{f}(-X)}{h(-X)} \ell^*(X) f^*(X)\\
    & = - \dfrac{  f(X) \ell(X) X^9 g(X)}{X^9f(X)}
     + \gamma \dfrac{X^9g(X) X^6 f(-X)}{X^9f(-X)} \ell(X) f(X)\\
    & = -\ell(X)g(X)+ \gamma g(X)X^6 \ell(X)f(X)\\
    & = (-X^5+1) + \gamma ( X^9 - X^8 + X^7 -X^6+ X^5-X^4 \\
    & \hspace{24pt}+X^3-X^2+X-1 )\in \C^{\perp_{TE}},
\end{align*}
and 
\begin{align*}
    \C^{\perp_{TE}} &= \big\langle    X^9 + 2X^5 + 2X^4 + 1,
     (-X^5+1) + \gamma ( X^9 - X^8 + X^7 \\
     & \hspace{24pt} -X^6+ X^5-X^4 +X^3-X^2+X-1 )   \big\rangle. 
\end{align*}
\end{example}

The following theorem provides generators for the trace Hermitian dual.

\begin{theorem} \label{hermitian dual skew cyclic}
    Let $w(X), \ell(X), f(X),g(X)\in \F_q[X]$ such that 
    \[X^n-1= w(X) \ell(X) f(X)g(X).\]  
   Consider the general $\F_q$-linear skew cyclic $\F_{q^2}$-code \[\C= \langle 
 w(X) f(X)+ \gamma q(X), \gamma w(X)g(X)\rangle .\]  
 Then 
\begin{align*}
     \C^{\perp_{TH}} = & \bigg\langle c'(X) \dfrac{X^n-1}{k(X)} + \gamma d'(-X) \dfrac{X^n-1}{k(-X)}\ell^*(X)f^*(X), \\ \nonumber
      & \qquad\qquad\qquad  \dfrac{\gamma^{q+1}  \ell^*(-X) f^*(-X) \widehat{q}(-X)}{h(X)}
     - \gamma \dfrac{\widehat{w}(-X)\widehat{f}(-X)}{h(-X)} \ell^*(X) f^*(X) \bigg\rangle,
 \end{align*}
where $h(X)= \gcd( \widehat{w}(X) \widehat{f}(X), \ell^*(-X)f^*(-X) \widehat{q}(-X))$, $k(X)=\gcd( h(X), X^n-1 )$, and $c'(X), d'(X)\in \F_q[X]$ are such that
\[   c'(X) \widehat{w}(X) \widehat{f}(X) +\gamma^{q+1} d'(X) \ell^*(-X) f^*(-X)\widehat{q}(-X) = h(X).  \]
Note here that $\gamma^{q+1} = -\gamma^2$.
 \end{theorem}

 \begin{proof}
Recall that $\gamma^{q+1} = -\gamma^2\in \F_q$ with $\tr (\gamma^{q+1}) = -2\gamma^{2}$, see Lemma \ref{gammasquare}. Otherwise, the proof is the same as for Theorem~\ref{main skew cyclic}.
 \end{proof}

 \section{Conclusion and Future Work}
In this article, we showed that assumptions used in \cite{Verma-Sharma:2024} are superfluous and we gave many direct proofs and generalizations for the results contained therein. In subsequent sections, these direct proofs allowed us to study the dual structure of $\F_{q}$-linear skew cyclic $\F_{q^2}$-codes under various inner products. We gave explicit formulas for the generators of these dual codes in terms of the parameters of the original code.

While we have limited the scope of the present article to $\F_q$-linear $\F_{q^2}$-codes, we are interested in extending the analysis of these codes to field extensions of different degrees. More study is also needed for a bound on the minimum distance of these nonlinear skew cyclic codes.

\bibliographystyle{plain}
\bibliography{references.bib}
\end{document}